\documentclass[USenglish, cleveref, autoref]{lipics-v2019}

\usepackage{microtype}
\usepackage{amsmath,array,bm}
\usepackage{amsmath,amsfonts}
\usepackage{amsthm}
\usepackage{amssymb,latexsym}
\usepackage{algorithm}
\usepackage{subcaption}
\usepackage{algorithmicx}
\usepackage[]{algpseudocode}
\usepackage[dvipsnames]{xcolor}
\usepackage{graphicx}
\usepackage{caption}
\usepackage{float}
\usepackage{thmtools}
\usepackage{hyperref}
\usepackage{mathtools}
\usepackage{cleveref}
\usepackage{bussproofs}
\usepackage{tikz}
\usepackage{xspace}
\usepackage{xcolor}
\usepackage{tipa}
\usepackage{mdframed}
\usepackage[export]{adjustbox}[2011/08/13]
\usepackage{multirow}
\usepackage{hhline}

\newtheorem{observation}[lemma]{Observation}

\title{A Constant Approximation for Colorful $k$-Center}
\titlerunning{A Constant Approximation for Colorful $k$-Center}
\author{Sayan Bandyapadhyay}{Department of Computer Science, University of Iowa, \\{Iowa City, IA, USA.}}{sayan-bandyapadhyay@uiowa.edu}{}{}
\author{Tanmay Inamdar}{Department of Computer Science, University of Iowa, \\{Iowa City, IA, USA.}}{tanmay-inamdar@uiowa.edu}{}{}
\author{Shreyas Pai}{Department of Computer Science, University of Iowa, \\{Iowa City, IA, USA.}}{shreyas-pai@uiowa.edu}{}{}
\author{Kasturi Varadarajan}{Department of Computer Science, University of Iowa, \\{Iowa City, IA, USA.}}{kasturi-varadarajan@uiowa.edu}{}{}
\authorrunning{S. Bandyapadhyay, T. Inamdar, S. Pai and K. Varadarajan}
\Copyright{Sayan Bandyapadhyay, Tanmay Inamdar, Shreyas Pai, and Kasturi Varadarajan}

\funding{The first, second and fourth authors are partially supported by the National Science Foundation under Grant CCF-1615845.}
\EventEditors{Michael A. Bender, Ola Svensson, and Grzegorz Herman}
\EventNoEds{3}
\EventLongTitle{27th Annual European Symposium on Algorithms (ESA 2019)}
\EventShortTitle{ESA 2019}
\EventAcronym{ESA}
\EventYear{2019}
\EventDate{September 9--11, 2019}
\EventLocation{Munich/Garching, Germany}
\EventLogo{}
\SeriesVolume{144}
\ArticleNo{10}

\hideLIPIcs

\newcommand{\LP}{\textsf{LP}\xspace}

\newcommand{\NP}{\textsf{NP}}
\nolinenumbers
\begin{document}
\maketitle
\begin{abstract}
	In this paper, we consider the colorful $k$-center problem, which is a generalization of the well-known $k$-center problem. Here, we are given red and blue points in a metric space, and a coverage requirement for each color. The goal is to find the smallest radius $\rho$, such that with $k$ balls of radius $\rho$, the desired number of points of each color can be covered. We obtain a constant approximation for this problem in the Euclidean plane. We obtain this result by combining a ``pseudo-approximation'' algorithm that works in any metric space, and an approximation algorithm that works for a special class of instances in the plane. The latter algorithm uses a novel connection to a certain matching problem in graphs.
\end{abstract}

\newpage
\pagenumbering{arabic}
\section{Introduction}

In the {\em $k$-center} problem, we are given a finite metric space $(P, d)$, where $P$ is a set of $n$ points, and $d: P \times P \rightarrow \mathbb{R}^+$ is the associated distance function. We are also given an integer $1 \le k \le n$. The goal is to find a subset $C \subseteq P$ of {\em centers}, where $|C| = k$, so as to minimize $ \max_{p \in P} \min_{c \in C} d(p, c)$, the maximum distance of a point from its nearest center. Geometrically, we want to find the smallest radius $\rho$ such that $P$ can be covered by $k$ balls of radius $\rho$ (centered at points in $P$). 

It is well known that $k$-center is \NP-hard; furthermore, it is also \NP-hard to approximate the optimal radius to within a $2-\epsilon$ factor, for any $\epsilon > 0$. This is easily seen via a reduction from the minimum dominating set problem \cite{hsu1979easy}. On the other hand, it is possible to obtain a tight approximation ratio of $2$ \cite{gonzalez1985clustering,hochbaumS1985best}. A simple greedy algorithm of \cite{gonzalez1985clustering} achieving this starts with $C$ containing an arbitrary point in $P$; in each of the subsequent $k-1$ iterations, it finds a point $p \in P$ that maximizes
$d(p, C) : = \min_{c \in C} d(p,c)$, and adds that point to $C$.

In the {\em $k$-center with outliers} problem, we are given an additional parameter $1 \le p \le n$. The goal is to find the smallest radius $\rho$ such that at least $p$ points of $P$ can be covered by $k$ balls of radius $\rho$. Thus, we allow up to $n - p$ points to remain uncovered, and these would be ``outliers.'' Intuitively, in comparison to $k$-center, an algorithm for solving this problem has to also figure out which $p$ points to cover. Nevertheless, Charikar et al.\ \cite{charikar2001algorithms}, who introduced this problem, gave a simple $3$-approximation using a greedy algorithm. Their algorithm guesses the optimal radius $\rho$; then, for $k$ iterations, it finds the ball of radius $\rho$ that covers the maximum number of yet uncovered points, expands it by a factor of $3$, and adds it to the solution. More recently, the approximation guarantee was improved to $2$, using \LP-rounding techniques \cite{chakrabarty2016non,harris2017lottery}. Note that this approximation guarantee is tight, in light of the $2-\epsilon$ hardness result mentioned earlier.

We now introduce the {\em colorful $k$-center} problem, which is studied in this article. As in $k$-center, the input consists of a set $P$ of $n$ points in a metric space, and an integer $k$. Furthermore, we are given a partition $\{P_1, P_2, \ldots, P_c \}$ of $P$ into $c$ {\em color classes}, and a coverage requirement $0 \leq t_i \leq |P_i|$ for each color class $1 \leq i \leq c$. The goal is to find the smallest radius $\rho$ such that using $k$ balls of radius $\rho$, centered at points of $P$, we can simultaneously cover at least $t_i$ points from each class $P_i$. When we have only one color class, that is, when $c = 1$, we obtain the $k$-center with outliers problem. In much of this article, we focus on the case with two color classes, where $c = 2$. In this case, we call the colors {\em red }and {\em blue}; we denote $P_1$ and $P_2$, respectively, by $R$ (red points) and $B$ (blue points), and denote the coverage requirements $t_1$ and $t_2$, respectively, by $r$ and $b$. The motivation for studying
this problem is as follows. Each color class can be thought of as a certain demographic with a specific coverage requirement, which must be satisfied by the $k$ balls chosen in the solution.

Even with two colors, the colorful $k$-center is quite challenging. The greedy algorithm for $k$-center with outliers \cite{charikar2001algorithms} has no obvious generalization: with two color classes, what do we optimize when adding the next ball to our solution? The \LP-based approaches for $k$-center with outliers do not generalize either -- as we point out (in Example \ref{ex:fractional}), the natural \LP has an unbounded integrality gap.

Chakrabarty et al.~\cite{chakrabarty2016non} study a closely related problem called non-uniform $k$ center. As input to this problem, we are given a set $P$ of points in a metric space, $\lambda$ distinct radii $r_1 > \cdots > r_\lambda \geq 0$ and corresponding integers $t_1 , \ldots, t_\lambda$. The goal is to find the smallest {\em dilation} $\beta \geq 0$ such that $P$ can be covered by a collection of balls formed by including, for each $1 \leq i \leq \lambda$, $t_i$ balls of radius $\beta \cdot r_i$. When $\lambda = 1$, we get the regular $k$-center problem. When $\lambda$ is unbounded, the problem is hard to approximate to within any constant factor; when $\lambda = 2$, one can get an $O(1)$-approximation; and when $\lambda$ is a constant greater than $2$, it is open as to whether an $O(1)$-approximation is possible \cite{chakrabarty2016non}.

As Chakrabarty et al.~\cite{chakrabarty2016non} observe, there is a close relationship between $k$-center with outliers and non-uniform $k$-center with $\lambda = 2$. In fact, it can be shown that the two problems are equivalent up to an $O(1)$-approximation factor. While the relationship between the colorful $k$-center with $c$ color classes and non-uniform $k$-center with $c + 1$ distinct radii is not known to be as close for $c \geq 2$, the study of the latter problem is one motivation that led us to the colorful $k$-center problem.


\subsection{Other Related Work}
The $k$-means and $k$-median are classic \NP-hard clustering problems that are closely related to the $k$-center problem. Like the $k$-center problem, these problems have been extensively studied, resulting in different approaches guaranteeing constant factor approximations. More recently, the outlier versions of these problems were also studied; constant factor approximations were obtained for $k$-median with outliers \cite{chen2008constant,krishnaswamy2018constant} and $k$-means with outliers \cite{krishnaswamy2018constant}. A polynomial time bicriteria $(1+\epsilon)$-approximation using at most $k(1+\epsilon)$ centers for any $\epsilon > 0$ is known in low dimensional Euclidean spaces, and metric spaces with constant doubling dimension \cite{friggstad2019approximation}.

Facility location with outliers, which is referred to as {\em Robust Facility Location}, is a generalization of the uncapacitated Facility Location problem; various constant approximations are known for the latter problem. The Robust Facility Location problem was introduced in \cite{charikar2001algorithms}, who gave a $3$-approximation. The approximation guarantee was later improved by Jain et al. \cite{jain2003greedy} to $2$.

A colorful version of vertex cover is studied in \cite{bera2014approximation}, and colorful versions of the Set Cover and Facility Location-type problems were considered in \cite{inamdar2018partition}. In these problems, the cardinality of the cover (or its weight) shows up in the objective function, unlike in $k$-center, where the number of centers/balls $k$ is a ``hard restriction''. These problems therefore have a different flavor.


Finally, $k$-center and $k$-median have been generalized in an orthogonal direction, where there are additional constraints on the centers \cite{hajiaghayi2010budgeted,chen2016matroid,chakrabarty2018generalized}. Again, the issues studied in these generalizations tend to be quite different from the ones we confront here.  


\subsection{Our Results}
We study the colorful $k$-center problem when the number of colors is a constant. Our main result is a randomized polynomial time algorithm that, with high probability, outputs an $O(1)$-approximation in the Euclidean plane. (As $k$-center is APX-hard even in the plane \cite{feder1988optimal}, we cannot hope for a PTAS.)

To describe our approach, we focus on the case with two colors. We first design a pseudo-approximation algorithm that outputs a $2$-approximation, but with $k + 1$ centers instead of $k$. This result holds in any metric space, not just the Euclidean plane. To obtain it, we preprocess the solution to a natural \LP-relaxation into a solution for a simpler \LP, using ideas from \cite{chakrabarty2016non,harris2017lottery}. We then note that a basic feasible solution to the simple \LP opens at most $k + 1$ centers fractionally. A pseudo-approximation for $k$-median with outliers was an important step in the recent work of Krishnaswamy et al. \cite{krishnaswamy2018constant}.

The pseudo-approximation allows us to reduce the colorful $k$-center problem to a special case where the balls in the optimal solution are separated -- the distance between any two balls is much greater than their radii. Designing an $O(1)$-approximation for this special case is challenging, even in the plane. For instance, partitioning into small subproblems by using a grid (such as in \cite{hochbaumM1985approximation}) or other type of object does not work, because any such partitioning may intersect balls in the optimal solution, and we have a hard bound $k$ on the number of balls allowed.

We solve separated instances in the plane by reducing to exact perfect matching on graphs. In this problem, we are given a graph in which each edge has a red weight and a blue weight, both non-negative integers. Given integers $w_r$ and $w_b$, the goal is to determine if the graph has a perfect matching whose red and blue weights are, respectively, {\em exactly} $w_r$ and $w_b$. This problem can be solved in randomized polynomial time, provided the weights are bounded by a polynomial in the input size; see  \cite{camerini1992random,mulmuley1987matching}. To our knowledge, this connection of geometric clustering and covering to exact matching is a novel one. 

In its current form, the reduction from separated instances of colorful $k$-center to exact matching does not work even in $\mathbb{R}^3$. Nevertheless, we are hopeful that this work will lead to an $O(1)$-approximation to colorful $k$-center in dimensions $3$ and higher, and indeed in any metric space.

\textbf{Organization.} We describe our pseudo-approximation in Section~\ref{sec:pseudo-appx} and our $O(1)$-approximation for the Euclidean plane in Section~\ref{sec:geometric}. In both these sections, we focus on the case where the number of colors is $2$. We address the extension of the planar result to multiple color classes in Section~\ref{sec:multiple}. 

\section{Pseudo-approximation via \LP Rounding}
\label{sec:pseudo-appx}

Recall that an instance of colorful $k$-center is a metric space $(P = R \sqcup B, d)$, where $R, B$ are non-empty, disjoint sets of red and blue points respectively. We are also given red and blue coverage requirements $1 \le r \le |R|$ and $1 \le b \le |B|$, respectively. A solution $(C, D)$, where $C, D \subseteq P$, is said to be feasible, if (i) $|C| \le k$, (ii) $|D \cap R| \ge r$, and (iii) $|D \cap B| \ge b$. The cost of a solution is defined as $\max_{j \in D} d(j, C)$. Any point in the set $D$ is said to be \emph{covered} by the solution. The goal of the colorful $k$-center problem is to find a feasible solution of the minimum cost.

In this section, we describe a pseudo-approximation algorithm for the colorful $k$-center problem. That is, we show how to find a solution of cost at most $2 \cdot OPT$ using at most $k+1$ centers, where $OPT$ is the cost of an optimal solution. This result is achieved in two steps. In the first step, we use the natural \LP relaxation for the decision version of the colorful $k$-center problem, to partition the points in $P$ into disjoint clusters using a simple greedy procedure. We also obtain a related fractional solution during this clustering procedure. We use this to show that a much simpler \LP with a small number of constraints has a feasible solution, even though there may not exist an integral feasible solution using at most $k$ centers. Nevertheless, we use the simplicity of the \LP to show that there exists a solution using at most $k+1$ centers. 

\begin{figure}
\begin{mdframed}[backgroundcolor=gray!9]
	Find a feasible solution $(x, z)$ such that:
	\begin{alignat}{3}
	\text{} \displaystyle\sum\limits_{i \in B(j, \rho)}   x_{i} &\geq z_j,  \quad &  \forall j \in P \label[constr]{constr:cover-ej}\\
	\displaystyle\sum_{i \in P}x_i &\le k, &  \label[constr]{constr:atmost-k}\\
	\displaystyle\sum_{j \in R}z_j &\ge r, & \label[constr]{constr:red-coverage}\\
	\displaystyle\sum_{j \in B}z_j &\ge b, & \label[constr]{constr:blue-coverage}\\
	\displaystyle z_j, x_i &\in [0, 1], \quad & \forall i, j \in P \label[constr]{constr:fractional}
	\end{alignat}
\end{mdframed}
\caption{The feasibility \LP, parameterized by $\rho$ \label{fig:FeasibilityLP}}
\end{figure}

Let $\rho$ denote our ``guess'' for the optimal cost. Note that the optimal cost must be one of the $O(n^2)$ interpoint distances, therefore there are $O(n^2)$ choices for $\rho$. We state the feasibility \LP, parameterized by $\rho$, in Figure \ref{fig:FeasibilityLP}.
It is easy to see that an optimal solution satisfies all the constraints when the guess $\rho$ is correct (i.e., when $\rho \ge OPT$). Therefore, henceforth, we assume that $\rho = OPT$. 

Now, we find a feasible fractional solution $(x', z')$ for this \LP, and use it to show that a related, but a much simpler \LP is feasible. To this end, we use the following ``greedy clustering'' procedure (see Algorithm \ref{alg:clustering}), which also computes a modified \LP solution $(\tilde{x}, \tilde{z})$. Let $P'$ denote the set of unclustered points, initialized to $P$. We also initialize $S$, the collection of cluster-centers, to the empty set. In each iteration, we find a point $j \in P'$ with the maximum $z_j$. We set $\tilde{z}_j = \tilde{x}_j \gets \min \{1, \sum_{i \in B(j, \rho)} x'_i \}$ -- note that the sum is over all points in the ball $B(j, \rho)$, as opposed to the points in $B(j, \rho) \cap P'$. Let $C_j$ denote the set of unclustered points within distance $2\rho$ from $j$. We refer to the set $C_j$ as a cluster. We set $\tilde{z}_{j'} \gets \tilde{z}_j$ for all other points $j' \in C_j \setminus \{j\}$. Finally, we remove the points in $C_j$ from $P'$ and repeat this process until all points are clustered, i.e., $P'$ becomes empty.

\begin{algorithm} 
	\caption{Clustering Algorithm} \label{alg:clustering}
	\begin{algorithmic}[1] 
		\State $S \gets \emptyset$, \quad $P' \gets P$
		\While{$P' \neq \emptyset$}
		\State $j \in P'$ be a point with maximum $z'_j$; let $S \gets S \cup \{j\}$
		\State $\tilde{x}_j \gets \min\{1, \sum_{i \in B(j, \rho)} x'_i\}$; $\tilde{z}_j \gets \tilde{x}_j$
		\State $C_j \gets B(j, 2\rho) \cap P'$
		\State For all $j' \neq j \in C_j$, set $\tilde{x}_{j'} \gets 0, \tilde{z}_{j'} \gets \tilde{z}_j$
		\State $P' \gets P' \setminus C_j$
		\EndWhile
	\end{algorithmic}
\end{algorithm}

For any point $i \in S$, let $R_i \coloneqq R \cap C_i$ and $B_i \coloneqq B \cap C_i$ denote the sets of red and blue points in the ``cluster'' $C_i$ respectively. Additionally, for any $i \in S$, let $r_i$ and $b_i$ denote the sizes of the sets $R_i$ and $B_i$ respectively. We start with a few simple observations that are immediate from the description of the procedure.

\begin{observation}\label{obs:clustering}\
	\begin{enumerate}
		\item The ``clusters'' $\{C_i\}_{i \in S}$ partition the point set $P$. Therefore, $\{R_i\}_{i \in S}$ partition the red points $R$, and $\{B_i\}_{i \in S}$ partition the blue points $B$,
		\item For any two distinct $i, i' \in S$, $d(i, i') > 2\rho$,
		\item For any $j \in P$, there is at most one $i \in S$, such that $d(i, j) \le \rho$.
	\end{enumerate}
\end{observation}

The following observation follows from the greedy nature of the clustering procedure.
\begin{observation} \label{claim:z-preservation}
	For any point $j_1 \in P$, let $j \in S$ be the point in $S$ such that $j_1 \in C_j$. Then, $\tilde{z}_j = \tilde{z}_{j_1} \ge z'_{j_1}$
\end{observation}
\begin{proof}
	Notice that, $\tilde{z}_{j_1} = \tilde{z}_j = \min\{1, \sum_{i \in B(j, \rho)} x'_{i} \}$. There are two cases to consider.
	
	$\tilde{z}_{j_1} = \tilde{z}_j = 1 \ge z'_{j_1}$, where the inequality follows from the constraint (\ref{constr:fractional}) of the \LP. 
	
	Otherwise, $\tilde{z}_j = \sum_{i \in B(j, \rho)} x'_i$. In this case, $\tilde{z}_{j} = \sum_{i \in B(j, \rho)} x'_i \ge z'_j \ge z'_{j_1}$. Here, the first inequality follows from constraint (\ref{constr:cover-ej}) of the \LP, and the second inequality follows from the fact that in the iteration when $j_1$ was removed from $P'$, $j \in P'$ was chosen to be a point with the maximum $z'$-value in line 3.
\end{proof}

The next two claims help us construct a feasible solution to a simplified \LP, to be introduced later.

\begin{claim}\label{claim:rb-coverage}\
	\begin{enumerate}
		\item $\sum_{i \in S} r_i \tilde{x}_i \ge r$
		\item $\sum_{i \in S} b_i \tilde{x}_i \ge b$
	\end{enumerate}
\end{claim}
\begin{proof}
	We prove the first part of the claim -- the second part is analogous.
	\begin{align*}
	\sum_{i \in S} r_i \tilde{x}_i &= \sum_{i \in S} |R_i| \cdot \tilde{x}_i
	\\&= \sum_{i \in S} \sum_{j' \in R_i} \tilde{z}_i \tag{For any $i \in S, \tilde{x}_i = \tilde{z}_i$ by construction}
	\\&\ge \sum_{i \in S} \sum_{j' \in R_i} z'_{j'} \tag{For any $j' \in R_i \subseteq C_i, \tilde{z}_i \ge z'_{j'}$ from Observation \ref{claim:z-preservation}}
	\\&= \sum_{j \in R} z'_j \tag{Property 1 of Observation \ref{obs:clustering}}
	\\&\ge r \tag{By constraint (\ref{constr:red-coverage})}
	\end{align*}
\end{proof}
\begin{claim}\label{claim:fractional-k}
	$\sum_{i \in S} \tilde{x}_i \le k$
\end{claim}
\begin{proof}
	\begin{align*}
	\sum_{i \in S} \tilde{x}_i &\le \sum_{i \in S} \sum_{i' \in B(i, \rho)} x'_{i'} \tag{$\tilde{x}_i \le \sum_{i' \in B(i, \rho)} x'_{i'}$}
	\\&\le \sum_{i' \in P} x'_{i'} \tag{From Property 3 of Observation \ref{obs:clustering}}
	\\&\le k \tag{By constraint (\ref{constr:atmost-k})}
	\end{align*}
\end{proof}

\subsection{A Simplified Problem}
Recall that the clusters $\{C_j\}_{j \in S}$ are disjoint, and have radius $2\rho$. Now, if we can find a collection of $k$ clusters from this set, that cover at least $r$ red points and $b$ blue points, then this immediately leads to a $2$-approximation. Unfortunately, it is not always possible to find such a collection. However, in the following, we show that we can find a collection of $k+1$ clusters that satisfies the coverage requirements of both colors. The \LP in Figure \ref{fig:SimplifiedLP} is a relaxation of the problem of finding at most $k$ clusters that satisfy the coverage requirements of both colors.

\begin{figure}
\begin{mdframed}[backgroundcolor=gray!9]
	\begin{alignat}{3}
	\text{maximize}\quad \displaystyle&\sum\limits_{i \in S} r_{i}x_{i} & \nonumber \\
	\text{subject to}\quad \displaystyle&\sum\limits_{i \in S} b_i x_{i} \geq b \label[constr]{constr:blue-cover1}\\
	&\displaystyle\sum_{i \in S}\ x_i\  \le k, \label[constr]{constr:atmost-k1}\\
	\displaystyle & \qquad x_i \in [0, 1], \quad & \forall i \in S \label[constr]{constr:fractional-x}
	\end{alignat}
\end{mdframed}
\caption{The Simplified \LP \label{fig:SimplifiedLP}}
\end{figure}

Note that Claims \ref{claim:rb-coverage} and \ref{claim:fractional-k} imply that the fractional solution $(\tilde{x})$ constructed above is a feasible solution for this \LP, and has objective value at least $r$. However, there may not exist a feasible integral solution that uses at most $k$ clusters from the set $\{C_i\}_{i \in S}$. The following example illustrates this.

\begin{example} \label{ex:fractional}
	In the following figure, we have two clusters $C_1$ and $C_2$. Red points are shown as boxes, whereas blue points are shown as dots. $C_1$ consists of $3$ red points and $1$ blue point, whereas $C_2$ contains $3$ blue points and $1$ red point. Suppose that $k = 1$ and the coverage requirements of each color class is $2$. Now, assigning $x_1 = x_2 = 0.5$ yields a fractional solution that satisfies the coverage requirements of each color class. However, assuming that the distance between two clusters is very large compared to the radii, $2\rho$, it can be seen that there is no feasible integral solution of cost at most a constant multiple of $2\rho$.
	\begin{figure}[H]
		\centering
		\includegraphics[scale=1]{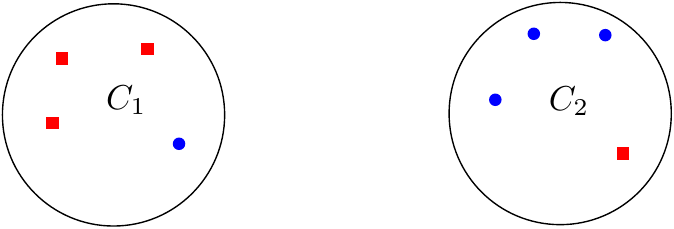}
	\end{figure}
\end{example}

Nevertheless, we show in the following that there exists a feasible solution that uses at most $k+1$ clusters. First, we need the following classical result from linear algebra (see Lemma 2.1.4 in \cite{lau2011iterative}).

\begin{lemma}[\cite{lau2011iterative}] \label{lemma:ranklemma}
	In any extreme point feasible solution (or equivalently, a basic feasible solution) to a linear program, the number of linearly independent tight constraints is equal to the number of variables.
\end{lemma}

Furthermore, an extreme point optimal solution can be computed in polynomial time. Now, we find such an optimal solution $(x^*)$ to the simplified \LP. It follows from Claim \ref{claim:rb-coverage} that its objective value is at least $r$. We prove the following lemma, which is a simple consequence of Lemma \ref{lemma:ranklemma}.

\begin{lemma} \label{lemma:pseudo-approx}
	The number of $x^*$-variables in an extreme point optimal solution that are strictly fractional, is at most $2$. Therefore, the number of strictly positive variables is at most $k+1$.
\end{lemma}
\begin{proof}
	Let $m \coloneqq |S|$ denote the number of variables. From Lemma \ref{lemma:ranklemma}, it follows that the number of linearly independent tight constraints is equal to $m$. Note that, even though there are $2m+2$ constraints in the \LP, at most $m+2$ constraints can be simultaneously tight. Now, even if constraints (\ref{constr:blue-cover1}) and (\ref{constr:atmost-k1}) are tight, it follows that number of tight constraints from $(\ref{constr:fractional-x})$ is at least $m-2$. That is, the number of strictly fractional variables is at most $2$.
	
	If there are $k$ variables that are equal to $1$, then constraint (\ref{constr:atmost-k1}) is tight, and there are no strictly fractional variables. Otherwise, the number of variables equal to $1$, is at most $k-1$. Along with at most $2$ strictly fractional variables, the number of positive variables is at most $k+1$.
\end{proof}
Note that the red and blue coverage can only increase while rounding up the fractional variables to $1$. If there is exactly one fractional variable, we round it up, giving a $2$-approximation using exactly $k$ centers. Otherwise, let $x^*_1, x^*_2$ be the two fractional variables, corresponding to clusters $C_1, C_2$ respectively. Note that, although the fractional solution satisfies both red and blue coverage, adding only $C_1$ or $C_2$ to the solution may not satisfy coverage demands for both colors -- recall Example \ref{ex:fractional}. Therefore, we include both clusters in the solution, resulting in a pseudo-approximation of cost at most $2\rho$, using at most $k+1$ centers. We summarize the result of this section in the following theorem.

\begin{theorem}
	There exists a polynomial time algorithm to find a $2$-approximation for the colorful $k$-center problem (with two colors) in any metric space, using at most $k+1$ centers.
\end{theorem}

This theorem generalizes readily to an arbitrary number of color classes $c \geq 2$. Having \(c\) color classes corresponds to having \(c-1\) constraints instead of constraint (\ref{constr:blue-cover1}) in the corresponding Simplified \LP. As we now have constraints corresponding to $c - 1$ color classes, we obtain the following lemma for the Simplified \LP, the proof of which goes along the same lines as that of Lemma \ref{lemma:pseudo-approx}.

\begin{lemma} \label{lemma:multi-pseudo-approx}
  The number of $x^*$-variables in an extreme point optimal solution that are strictly fractional, is at most $c$. Therefore, the number of strictly positive variables is at most $k+c-1$.
\end{lemma}

Therefore, if we round up these fractional variables, we get a pseudo-approximation of cost at most \(2 \cdot \rho \), using at most \(k+c-1\) centers.

\section{A Constant Approximation Algorithm in \(\mathbb{R}^2\)} \label{sec:geometric}
In this section, we describe a constant approximation algorithm for the colorful \(k\)-center problem with two color classes, red and blue, where the set of points lies in the Euclidean plane. Recall that $P = R \sqcup B$ denotes the input set of $n$ points, and $r$ and $b$ the red and blue coverage requirements. Our algorithm consists of two subroutines which we describe in the following two subsections. The two subroutines are intended to handle two different types of instances. In order to describe these instances we need the following definitions.
\begin{definition}
  Let $\alpha > 0$ be a parameter. A set \(S\) of balls of radius \(\rho'\) is \(\alpha\)-separated if the distance between the centers of every two balls in \(S\) is greater than \(\alpha\cdot\rho'\). An instance of colorful $k$-center is \(\alpha\)-separated if the set of balls in some optimal solution is \(\alpha\)-separated.
\end{definition}

The first subroutine (described in Section \ref{sec:non-separated}) gives a \(2 (\alpha + 1)\)-approximation algorithm for instances that are not \(\alpha\)-separated. The second subroutine (described in Section \ref{sec:perfect-matching}) gives a \((0.5 \alpha + 2)\)-approximation algorithm for \(\alpha\)-separated instances where \(\alpha > 4/\gamma\) for some absolute constant \(\gamma > 0\), which is defined below.


Therefore for a large enough constant \(\alpha\), we get a \(\max\{2 (\alpha + 1), 0.5 \alpha + 2\}\)-approximation algorithm for the colorful \(k\)-center problem with two color classes. From the geometric arguments in Section \ref{sec:perfect-matching} that determine $\alpha$, it is apparent that taking $\alpha = 8+\epsilon'$ for any $\epsilon' > 0$ is sufficient. Therefore, the approximation guarantee of our algorithm is $17+\epsilon$ for any $\epsilon > 0$. 

\subsection{Handling Non-Separated Instances}
\label{sec:non-separated}
If an instance is not \(\alpha\)-separated, we can use the pseudo-approximation algorithm of Section \ref{sec:pseudo-appx} to immediately get a \(2 (\alpha + 1)\)-approximation algorithm. This is formalized in the following lemma. It is worth pointing out that this subroutine does not require the set of points to be in \(\mathbb{R}^2\) -- it works for any metric.

\begin{lemma} \label{lemma:separability}
  If an instance is not \(\alpha\)-separated then we get a \(2 (\alpha + 1)\)-approximate solution to the colorful \(k\)-center problem in polynomial time.
\end{lemma}
\begin{proof}
  Let the optimal radius be \(\rho\). Since the instance is not \(\alpha\)-separated, there are two balls in some optimal solution whose centers are within distance \(\alpha \cdot \rho\) of each other. Let \(C_1\) and \(C_2\) be two such balls in the optimal solution. We replace \(C_1\) with a ball of radius \((\alpha + 1)\rho\) centered at the same point as \(C_1\). This allows us to remove \(C_2\) without violating any of the coverage requirements (since the new ball replacing \(C_1\) covers all points originally covered by \(C_2\)). This means that there exists a feasible solution using \(k-1\) centers with cost \((\alpha + 1)\rho\). Therefore, if we run the pseudo-approximation algorithm of the previous section with number of centers being \(k-1\), we will get a solution using at most \(k\) centers having cost \(2 (\alpha + 1)\rho\), which proves the lemma.
\end{proof}

\subsection{Reduction of Separated Instances to Exact Perfect Matching}
\label{sec:perfect-matching}
Let us assume that the instance we are given is \(\alpha\)-separated for some \(\alpha > 4/\gamma\) for some absolute constant \(\gamma > 0\). From the proof of Lemma \ref{lemma:extended-Voronoi-intersection}, it will be clear that taking $\alpha > 8$ is sufficient. The $\alpha$-separability of the instance helps us design a \((0.5\alpha + 2)\)-approximation algorithm for this problem in \(\mathbb{R}^2\). We do this by reducing the problem to the \emph{Exact Perfect Matching} problem \cite{papadimitriou1982complexity,berger2011budgeted}. In the exact perfect matching problem we are given an edge-weighted graph \(G = (V, E)\) and a target weight \(W\). The goal is to find a perfect matching in \(G\) having weight \emph{exactly} \(W\). The result in \cite{camerini1992random,mulmuley1987matching} gives a randomized pseudo-polynomial algorithm for exact perfect matching. In other words, their algorithm runs in polynomial time if the largest edge weight in the input graph is bounded by a polynomial in \(|V|\). We now describe our reduction from the $k$-center problem to exact perfect matching.

We first assume that we have guessed correctly the radius $\rho$ of an optimal solution to the given instance. We cover \(\mathbb{R}^2\) with a grid of equilateral triangles of side length \(\ell = 0.5 \alpha \rho\). See Figure \ref{fig:grid-graph} for an illustration. Consider the following three lines through the origin: \(L_1\) is the \(x\)-axis, \(L_2\) has angle \(60\) degrees with \(L_1\), and \(L_3\) has angle \(120\) degrees with \(L_1\). The grid of side length \(\ell\) can be formally defined as the arrangement of the collection of lines parallel to \(L_1\), \(L_2\), and \(L_3\), with two adjacent parallel lines having distance \(\sqrt{3}\ell/2\). We can interpret the grid as an infinite graph where the edges are the sides of each atomic triangle in the grid and the vertices are the intersection points of the edges.

Define the \emph{Voronoi region of an edge} \(e\) as the set of points that are at least as close to \(e\) as any other edge. It is easy to see that the Voronoi region of \(e\) is a rhombus whose four end points are the two end points of \(e\) and the centroids of the two triangles sharing the edge \(e\). We are interested in the \emph{extended Voronoi region of an edge} \(e\) which is the Minkowski sum of the Voronoi region of \(e\) and a ball of radius \(\rho\) centered at the origin. 


\begin{figure}[h]
  \centering
  \includegraphics{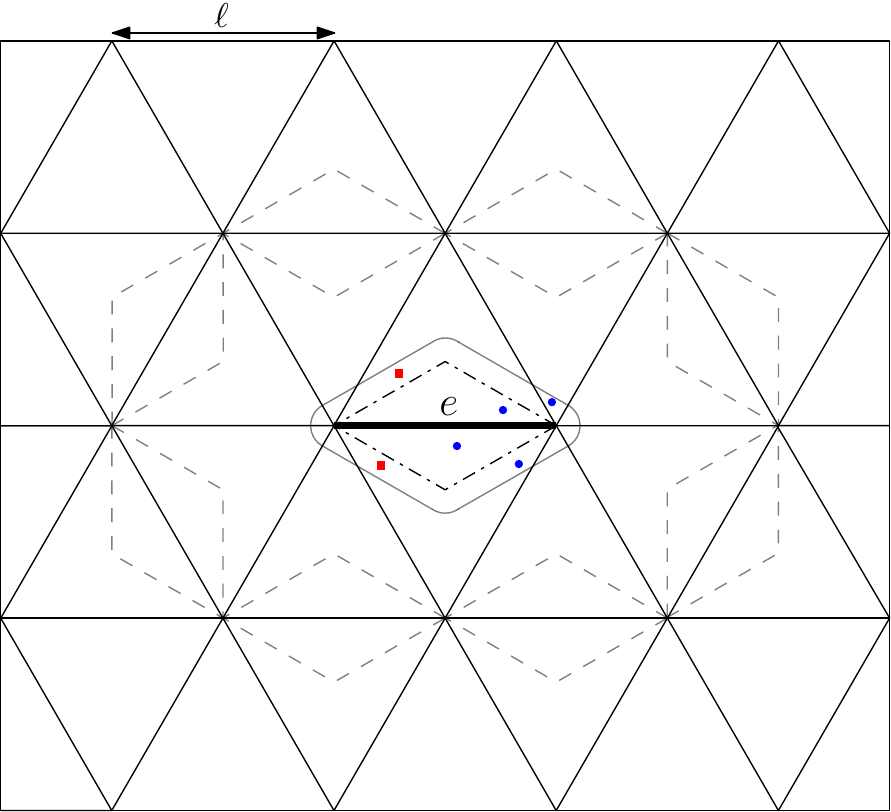}
  \caption{Here, we show a triangular grid of side $\ell = 0.5\alpha \rho$. We also show the Voronoi region of the edge $e$ (the area inside the dash-dotted rhombus around $e$), and the extended Voronoi region of $e$ (the area inside the lightly shaded object). We also show the Voronoi regions of the ``neighboring edges'' that do not share a vertex with $e$ as dashed rhombuses. Notice that the Voronoi regions (resp.\ extended Voronoi regions, not shown for simplicity) of the neighboring edges do not intersect with the Voronoi region (resp.\ extended Voronoi region) of $e$. We use this fact in the proof of Lemma \ref{lemma:extended-Voronoi-intersection}. Notice that, the extended Voronoi region of $e$ contains $2$ red points and $4$ blue points. Therefore, we will set the weight of the edge to be $4 \cdot N^2 + 2$. \label{fig:grid-graph}}
\end{figure}

\begin{lemma}\label{lemma:extended-Voronoi-intersection}
  The extended Voronoi regions of two edges \(e\) and \(e'\) intersect each other iff \(e\) and \(e'\) share a common vertex.
\end{lemma}
\begin{proof}
  The reverse direction of the lemma is trivial, so we focus on the forward direction. Take two edges \(e\) and \(e'\) that are not incident on a common vertex. Note that their Voronoi regions do not intersect (see Figure \ref{fig:grid-graph}). Thus, there is an absolute constant \(\gamma > 0\) such that the distance between the closest pair of points, one in each region, is at least \(\gamma \cdot \ell = \gamma \cdot 0.5 \alpha \rho\). Therefore, if \(\alpha > 4/\gamma\), the distance between the two Voronoi regions is greater than \( 2 \rho\). \footnote{With simple geometric calculations, one can show that $\gamma = 1/2$, which implies that $\alpha > 8$.} Thus, the extended Voronoi regions of \(e\) and \(e'\) do not intersect because each Voronoi region ``expands'' by only an additive factor of \(\rho\).
\end{proof}

For each edge define \(P(e)\) as the set of points in \(P\) that lie in the extended Voronoi region of \(e\). Let \(G'=(V', E')\) be the graph where \(E'\) is the set of all grid edges such that \(P(e) \neq \emptyset\), and \(V'\) is the set consisting of the endpoints of edges in \(E'\). Define \(N = 2\max\{n, |V'|\}\), note that \(N = O(n)\) because each point can belong to at most six triangles. The weight of an edge \(e \in E'\) is set to \(w_e = N^2 \cdot b_e + r_e\) where \(b_e\) and \(r_e\) are the number of blue and red points respectively in \(P(e)\).

\begin{lemma}\label{lemma:opt-is-matching}
  There exists a matching in \(G'\) with \(k\) edges and with weight exactly \(b' \cdot N^2 + r'\) for some \(|B| \geq b' \ge b\) and \(|R| \geq r' \ge r\).
\end{lemma}
\begin{proof}
  Let \(S^*\) be the set of balls in an optimal solution that is \(\alpha\)-separated. For each ball \(C_i \in S^*\), \(1 \le i \le k\), let \(e_i\) be any edge (from the infinite grid) such that the center of \(C_i\) lies in the Voronoi region of \(e_i\). Note that the extended Voronoi region of $e_i$ contains $C_i$, and thus $e_i \in E'$.  We claim that if $i \neq j$, then (a) $e_i \neq e_j$, and (b) $e_i$ and $e_j$ are not incident on a common vertex of $V'$. If either (a) or (b) does not hold, the distance between the centers of $C_i$ and $C_j$ is at most \(0.5 \alpha \rho + 0.5 \alpha \rho = \alpha \rho\), which violates the fact that the instance is \(\alpha\)-separated. Therefore the set \(M = \{e_i | 1 \le i \le k\}\) is a matching of \(k\) edges in \(G\).

  If we set \(b'\) and \(r'\) to be the total number of blue and red points in the extended Voronoi regions of the \(k\) edges, then we have a matching of weight \(b' \cdot N^2 + r'\) in \(G\). Since the extended Voronoi regions cover all the points covered by \(S^*\), we have \(b' \ge b\) and \(r' \ge r\) which finishes the proof.
\end{proof}

We now define a graph \(G\) such that perfect matchings in $G$ correspond to matchings with $k$ edges in $G'$. To construct \(G = (V, E)\), we start with \(G'\) and add \(|V'| - 2k\) auxiliary vertices that are connected to all vertices of \(G'\) with edges of weight zero. We can do this because Lemma \ref{lemma:opt-is-matching} implies that \(|V'| \ge 2k\). Now \(G\) obeys the following property that is easy to prove.

\begin{observation}\label{obs:perfect-matching}
  Any matching in \(G'\) of weight \(W\) having exactly \(k\) edges can be extended to a perfect matching in \(G\) of weight exactly \(W\). Conversely, from any perfect matching in $G$ with weight exactly $W$, we can obtain a matching in $G'$ with $k$ edges and weight exactly $W$.
\end{observation}

Using Lemma \ref{lemma:opt-is-matching} and Observation \ref{obs:perfect-matching} we get the following Corollary.

\begin{corollary}\label{cor:opt-is-matching}
  There exists a perfect matching in \(G\) of weight exactly \(b' \cdot N^2 + r'\) for some \(|B| \geq b' \ge b\) and \(|R| \geq r' \ge r\).
\end{corollary}

We now show how to go from a perfect matching in $G$ to a cover for the input instance.

\begin{lemma} \label{lemma:matching-to-kcenter}
  If \(G\) has a perfect matching of weight exactly \(W = b' N^2 + r'\), with $0 \leq b' \leq |B|$ and $0 \leq r' \leq |R|$, then we can place exactly \(k\) balls of radius \((0.5\alpha + 2)\rho\) that cover at least \(r'\) red points and at least \(b'\) blue points.
\end{lemma}
\begin{proof}

  Suppose $G$ has a perfect matching $M$ of weight $W$. Using Observation~\ref{obs:perfect-matching}, we recover a matching $M'$ in $G'$ with $k$ edges and weight $W$. For $e \in M'$, let $b_e = |P(e) \cap B|$ and $r_e = |P(e) \cap R|$. Thus, the weight \(w_e\) for each edge \(e\) is \(b_{e}N^2 + r_e\). It follows that
  \[ \left( \sum_{e \in M'} b_e \right) N^2 + \left( \sum_{e \in M'} r_e \right) = W = b' N^2 + r'.\]
  By our choice of $N$, both $r'$ and $\sum_{e \in M'} r_e$ are strictly less than $N^2$. It follows that
  $\sum_{e \in M'} r_e = r' = W \text{ mod } N^2$, and furthermore $\sum_{e \in M'} b_e = b'$. Furthermore, since $M'$ is a matching, by Lemma \ref{lemma:extended-Voronoi-intersection}, the extended Voronoi regions of the edges in $M'$ do not intersect. Thus, \(\cup_{e \in M'} P(e)\) contains exactly \(r'\) red points and \(b'\) blue points. 
  

  For each $e \in M'$ such that $P(e) \neq \varnothing$, we place a ball of radius \((0.5\alpha + 2) \rho\) centered at any point in \(P(e)\). These at most $k$ balls cover all points in \(\bigcup_{e \in M'} P(e)\). This is because the distance between any two points in \(P(e)\) for an edge \(e\) is at most \((0.5\alpha + 2) \rho\). This finishes the proof of the lemma.
\end{proof}

We are now ready to state our algorithm for our $\alpha$-separated instance of colorful $k$-center, assuming that we know the optimal radius $\rho$.  We construct the graph $G$ and check if it has a perfect matching of weight exactly $W = b' \cdot N^2 + r'$, for each $b \leq b' \leq |B|$ and $r \leq r' \leq |R|$. For the check, we invoke the algorithm of \cite{camerini1992random}.
Corollary \ref{cor:opt-is-matching} ensures that for at least one of these guesses for \(W\), an exact perfect matching does exist, and the algorithm of \cite{camerini1992random} returns it.\footnote{The algorithm of \cite{camerini1992random} is actually a randomized Monte-Carlo algorithm, and thus has an error probability that can be made arbitrarily small. For the sake of exposition, we ignore this eventuality of error, except in the statement of our final result.} Once we find an exact perfect matching solution for any of these weight values, we convert it into a \((0.5\alpha + 2)\)-approximate solution by invoking the algorithm of Lemma \ref{lemma:matching-to-kcenter}. This completes our algorithm.  


As we do not know the optimal radius for the given input instance, we will run the above algorithm for each choice of $\rho$ from the set of $O(n^2)$ interpoint distances induced by $P$. Fix a particular choice of $\rho$. If there exists a feasible solution with balls of radius $\rho$ that is $\alpha$-separated, then the above algorithm will return a feasible solution with cost at most   \((0.5\alpha + 2) \cdot \rho\). If there is no feasible solution with balls of radius $\rho$ that is $\alpha$-separated, then the above algorithm may not return a feasible cover.

Overall, we simply return the minimum cost feasible solution that we encounter after applying the above algorithm for all choices of $\rho$. This is a \((0.5\alpha + 2)\)-approximation for $\alpha$-separated instances.

\begin{remark}
  We note that the algorithm described here generalizes to any constant number of color classes, in the following sense. Suppose we have an input instance for which there exists an $\alpha$-separated solution of cost $\rho'$, for some $\rho'$. In particular, we do not assume that $\rho'$ is the optimal cost for the input instance. Then the algorithm computes a feasible solution of cost at most   \((0.5\alpha + 2) \cdot \rho'\). Note that in this case, for any edge $e$ of the graph $G'$, the weight of each edge will be of the form \(\sum_{i\ge 0} P_i(e) N^{2i}\) where \(P_i(e)\) is the number of points of color class \(i\) in the extended Voronoi region of that edge.
\end{remark}

\subsection{Combining Separated and Non-Separated Cases} 
We do not know if a given instance of colorful $k$-center is \(\alpha\)-separated or not. Therefore, we do not know which subroutine to use. So we run both the subroutines and return the solution with the lowest cost. We conclude with our main result.

\begin{theorem}
  There is a randomized polynomial time algorithm that, given any instance of colorful $k$-center in the plane with two colors, outputs, with high probability, an $O(1)$-approximate solution.
\end{theorem}  

\section{Multiple Color Classes in the Plane}
\label{sec:multiple}
So far we have only considered two color classes for the sake of keeping the exposition simple. We now sketch an extension to the colorful $k$-center problem in the plane with \(c\) color classes for any integer constant $c$. 




For a given instance in the plane, fix the optimal set of balls $S^*$, and assume it has cost $\rho$. Let $S$ be any maximal $\alpha$-separated subset of $S^*$. That is, for each pair of balls in $S$, the distance between the centers is greater than $\alpha \cdot \rho$, whereas for every ball $B \in S^* \setminus S$, there  is a ball $B' \in S$, such that the distance between the centers of $B$ and $B'$ is at most $\alpha \cdot \rho$. Let $\bar{S} = S^* \setminus S$. We consider two cases:
\paragraph*{Case 1: $|\bar{S}| \geq c - 1$} Notice that if we expand each ball in $S$ to have radius $(\alpha + 1) \cdot \rho$, and discard the balls in $\bar{S}$, we obtain a feasible solution with at most $k - (c-1)$ balls and cost $(\alpha + 1) \cdot \rho$. For this case, running the pseudo-approximation algorithm on the original input, but with the number of centers set to $k - (c-1)$, will give us a solution with $k$ balls and cost $2 (\alpha + 1) \cdot \rho$.

\paragraph*{Case 2: $| \bar{S}| \leq c-2$} We guess the balls in $\bar{S}$. Because $|\bar{S}| \leq c- 2$, and there are $O(n^2)$ choices for $\rho$, we have $O(n^c)$ possibilities to guess from. After guessing $\bar{S}$, we remove the points covered by $\bar{S}$, and decrease the coverage requirement for each color class by the number of points of that color class covered by $\bar{S}$. For this residual instance, $S$ is an $\alpha$-separated solution with $k - |\bar{S}|$ balls and cost $\rho$. If we run the algorithm of Section~\ref{sec:perfect-matching} on this residual instance, allowing $k - |\bar{S}|$ centers, we obtain a cover with $k - |\bar{S}|$ balls and cost at most $(0.5\alpha + 2)\cdot \rho$. We output
the union of this cover and $\bar{S}$. This is a feasible solution to the original problem with $k$ balls and cost $O(\rho)$.

The overall running time of the algorithm is $n^{O(c)}$.

\section{Open Problems}
One question raised by our work is whether an $O(1)$-approximation for colorful $k$-center (with a constant number of colors) can be obtained in $\mathbb{R}^d$ for $d \geq 3$, or indeed in an arbitrary metric space. Another question is whether one can obtain an $O(1)$-approximation for non-uniform $k$-center in the plane (with a constant number of distinct radii). 
\bibliography{references}
\end{document}